\documentclass{amsart}

\newif\ifdraft\draftfalse
\newif\ifcite\citefalse
\newif\ifblow\blowtrue

\usepackage{amsfonts,amssymb, amsmath}
\usepackage{oldgerm}
\usepackage{amscd}
\usepackage{eucal} 
\usepackage{mathrsfs} 
\usepackage[hyperindex]{hyperref}  
\usepackage[alphabetic]{amsrefs}
\usepackage{mathdots}
\ifdraft\ifcite\usepackage{showkeys}\else\usepackage[notcite,notref]{showkeys}\fi\fi

\usepackage{xypic}
\newdir{ >}{{}*!/-5pt/\dir{>}}

\newtheorem{proposition}[equation]{Proposition}
\newtheorem{theorem}[equation]{Theorem}

\newtheorem{lemma}[equation]{Lemma}

\theoremstyle{remark}

\theoremstyle{definition}

\theoremstyle{remark}

\newtheorem{remark}[equation]{Remark}

\newtheorem{rem}[equation]{Remark}

\numberwithin{equation}{section}

\def\bc{\begin{cases}}
\def\ec{\end{cases}}

\def\a{\alpha}

\def\t{\tilde}

\def\ci{{\mathcal I}}

\def\bc{{\mathbb C}}

\def\er{\eqref}

\def\bc{\mathbb C}

\def\lp2{L_pH_{2p}}
\def\bean{\begin{eqnarray}}
\def\eean{\end{eqnarray}}
\def\bea{\begin{eqnarray*}}
\def\eea{\end{eqnarray*}}
\def\beq{\begin{equation}}
\def\eeq{\end{equation}}
\def\beq*{\begin{equation*}}
\def\eeq*{\end{equation*}}
\def\bal{\begin{align*}}
\def\eal{\end{align*}}
\def\baln{\begin{align}}
\def\ealn{\end{align}}

\def\beg{\begin{gather*}}
\def\eng{\end{gather*}}
\def\bqu{\begin{question}}
\def\equ{\end{question}}

\def\implies{\Longrightarrow}

\def\nm{\nonumber}

\def\ban{\begin{proof}[Answer]}
\def\ean{\end{proof}}

\def\p{\partial}

\def\bqu{\begin{question}}
\def\equ{\end{question}}

\def\0110{\begin{matrix} 0 & 1\\1&0\end{matrix}}

\def\t{\tilde}

\def\fg{\mathfrak{g}}
\def\fh{\mathfrak{h}}

\def\fl{\mathfrak{l}}

\def\fn{\mathfrak{n}}
\def\fo{\mathfrak{o}}
\def\fp{\mathfrak{p}}

\def\fs{\mathfrak{s}}

\def\ban{\begin{proof}[Answer]}
\def\ean{\end{proof}}

\def\ben{\begin{equation}}
\def\een{\end{equation}}

\def\j1{{(j+1)}}

\def\e{\epsilon}

\begin{document}

\title[Solutions to Toda theories and related properties]
{
Solving Toda field theories and related 
algebraic and differential properties
}

\author{Zhaohu Nie}
\email{zhaohu.nie@usu.edu}

\address{Department of Mathematics and Statistics, Utah State University, Logan, UT 84322-3900}



\begin{abstract}
Toda field theories are important integrable systems. They can be regarded as constrained WZNW models, and this viewpoint helps to give their explicit general solutions, especially when a Drinfeld-Sokolov gauge is used. The main objective of this paper is to carry out this approach of solving the Toda field theories for the classical Lie algebras, following \cite{W-sym}. In this process, we discover and prove some algebraic identities for principal minors of special matrices. The known elegant solutions of \cite{L} fit in our scheme in the sense that they are the general solutions to our conditions discovered in this solving process. To prove this, we find and prove some differential identities for iterated integrals. It can be said that altogether our paper gives complete mathematical proofs for Leznov's solutions. 
\end{abstract}

\maketitle

\section{Introduction}
The Liouville equation is
\begin{equation}\label{liouville}
u_{xy}=-e^{2u},
\end{equation}
where $x$ and $y$ are the independent variables, and $u$ is an unknown function. 
Liouville found the general solutions to \er{liouville} involving two arbitrary functions $f(x)$ and $g(y)$, which depend on the two independent variables separately. (For this and general backgrounds on integrable systems, we refer to \cite{BBT}.) 

The Toda field theories are generalizations of the Liouville equation in the following way. Let $\fg$ be a complex simple Lie algebra of rank $n$. Let $\fh\subset \fg$ be a Cartan subalgebra, and we denote the corresponding set of roots of $\fg$ by $\Delta$, the sets of positive/negative roots by $\Delta_\pm$, and the set of positive simple roots by $
\{\a_i\}_{i=1}^n$. Let $\fg=\fh\oplus \bigoplus_{\a\in \Delta} \fg_\a$ be the root space decomposition. For $\a\in \Delta_+$, let $e_\a$ and $e_{-\a}$ be root vectors in the root spaces $\fg_\a$ and $\fg_{-\a}$ such that for $H_\a=[e_\a,e_{-\a}]\in \fh$, we have $\a(H_\a)=2$. 
Then the Cartan matrix of $\fg$ is 
\begin{equation}\label{Cartan}
A=(a_{ij})_{i,j=1}^n
\end{equation}
defined by $a_{ij}=\a_i(H_{\a_j})$.
For $1\leq i\leq n$, let $u_i$ be $n$ unknown functions of the independent variables $x$ and $y$. The Toda field theory associated to $\fg$ 
is 
\begin{equation}\label{toda}
u_{i,xy}=-e^{\rho_i}:=-\exp\Big({\sum_{j=1}^n a_{ij}u_j}\Big),\quad 1\leq i\leq n.
\end{equation}
(More specifically, these are the \emph{conformal} Toda field theories, as opposed to the affine ones which we don't consider in this paper.) 

The Cartan matrix $A$ \er{Cartan} completely determines the complex simple Lie algebra $\fg$. 
The classification of simple Lie algebras (see for example \cite{FH}) asserts that they come in four infinite series, listed below:
\begin{align*}
A_n&=\fs\fl_{n+1}, \quad n\geq 1, &&B_n=\fs\fo_{2n+1},\quad n\geq 2,\\
C_n&=\fs\fp_{2n},\quad n\geq 3, &&D_n=\fs\fo_{2n},\quad n\geq 4,
\end{align*}
and finitely many exceptional ones. 
We will also use the same letters for the corresponding Lie groups, where in the orthogonal cases we allow the determinants to be $\pm 1$. For example, $C_2=Sp(4,\bc)$ and $B_3=O(7,\bc)$. The Liouville equation \er{liouville} is the Toda field theory associated to $A_1$.

The Toda field theories \er{toda} admit a zero-curvature presentation \cite{LS} as follows. Define
\begin{equation}\label{def eps}
\epsilon:=\sum_{i=1}^n e_{-\a_i}.
\end{equation}
Then \er{toda} is equivalent to 
\begin{equation}\label{zero-curv}
\Big[\p_x-\sum_{i=1}^n u_{i,x} H_{\a_i}-\epsilon,\p_y+\sum_{i=1}^n e^{\rho_i} e_{\a_i}\Big]=0. 
\end{equation}
There had been a lot of studies on how to use this zero-curvature presentation to solve the Toda field theory, going back to \cite{LS}. (See \cite{LS-book} for more details.) 

In this paper, we follow \cite{W-sym} to investigate the solutions to the Toda field theories. 
It was shown in \cite{feher-early} that that \er{toda} can be regarded as a constrained WZNW model for conformal field theory associated to the Lie algebra $\fg$. The general solution to the WZNW model is the product of two chiral fields
\begin{equation}\label{chirality}
\Upsilon(x,y)=\Phi(x)\cdot \Psi(y),
\end{equation}
where $\Phi(x)$ and $\Psi(y)$ take values in the group $G$ corresponding to $\fg$. 
To get the Toda field theory, \cite{feher-early} puts on the following constrains. On the $x$-side, one has 
\begin{equation}\label{reduction}
\begin{split}
J:&=\p_x \Phi\cdot \Phi^{-1}\in \fg,\\
\pi_-J&=\epsilon=\sum_{i=1}^n e_{-\a_i},
\end{split}
\end{equation}
where $\pi_\mp:\fg\to \fn_\mp:=\oplus_{\a\in \Delta_\mp}\fg_\a$ are the canonical projections. 
On the $y$-side, one has
\begin{equation}\label{reduction-y}
\begin{split}
K:&=\Psi^{-1}\cdot \p_y \Psi \in \fg,\\
\pi_+ K&=\sum_{i=1}^n e_{\a_i}.
\end{split}
\end{equation}
In our classical Lie algebras, the root vector $e_{\a}$ as a matrix is 
the transpose of the matrix $e_{-\a}$. So we will concentrate on studying \er{reduction}, and the solutions to \er{reduction-y} are the solutions to \er{reduction} transposed and with $x$ replaced by $y$. 

The most convenient way to solve $\Phi(x)$ in \er{reduction} is through a Drinfeld-Sokolov gauge (DS-gauge for short) \cite{DS} as was done in \cite{W-sym} in the $A_n$ case. This gauge is related to the ``slice" of Kostant \cite{K1} for invariant functions on Lie algebras. Let $\fs$ be a complement of $[\e,\fg]$ in $\fg$, that is, $\fg\cong \fs\oplus [\e,\fg]$. 
Then $\fs\subset \fn_+=\bigoplus_{\a\in \Delta_+}\fg_\a$, and $\dim(\fs)=n$ is equal to the rank. Let $\{s_j\}_{j=1}^n$ be a homogeneous basis of $\fs$ with respect to the height grading. (We refer to \cite{K1} for more details.) 
The $J$ in \er{reduction} can be gauge transformed into $\e+\sum_{j=1}^n  I_j(x)s_j$ for $n$ functions $ I_j(x)$. Then \er{reduction} becomes 
\begin{equation}\label{find phi}
\p_x \Phi\cdot \Phi^{-1}=\e+\sum_{j=1}^n  I_j(x)s_j. 
\end{equation}
This equation can be solved in terms of some functions of $x$ satisfying some natural conditions, which are systems of ordinary differential equations. 

On the $y$-side, we have the transposed version. Then the solutions  $u_i(x,y)$ to \er{toda} is obtained from $\Upsilon(x,y)$ in \er{chirality} by the means of some principal minors, taking into consideration the residual gauges as in \cite{W-sym}. 

This process is carried out for the $A_n$ case in \cite{W-sym}, and there is a 
Wronskian condition for the solutions to \er{find phi}.
(See Theorem \ref{thm-A_n} Part (1).)
In Section \ref{solve}, we carry this process out for the $C_n$ and $B_n$ cases, and we find the conditions in Parts (1) of our main Theorems \ref{thm-C_n} and \ref{thm-B_n} for the solutions to \er{find phi}. We then discuss in Section \ref{solve} how to find the solutions $u_i(x,y)$ to \er{toda} using principal minors. 

After this search for solutions is done, we turn to directly proving that they are indeed solutions. 
In this process, we discover and prove some algebraic properties of minors for general linear, symplectic and orthogonal matrices as presented in Section \ref{ind int}. 
It is such algebraic identities that enable one to solve the $A_n$ Toda field theories and to view the $C_n$ and $B_n$ cases as reductions of the $A_n$ case. 

\begin{remark} We comment that the $ I_j(x)$ in \er{find phi} are also called \emph{local conservation laws} of the Toda field theory or \emph{intermediate integrals} from the viewpoint of Darboux integrability. For a very explicit presentation of the $ I_j(x)$ in terms of the Toda fields $u_i(x,y)$, see for example \cite{N1}. Such quantities, of course, are not unique, and there is a whole differential algebra of them, that is, any polynomial of the conserved quantities and their derivatives is still conserved. Later we will see some other generators of the local conservation laws in the $C_n$ and the $B_n$ cases, when we apply the Gram-Schmidt process. See \er{def I_j} and \er{def I_j-B}.  
\end{remark}

Leznov \cite{L} actually took things one step further by solving the Wronskian condition for the $A_n$ case using iterated integrals of $n$ arbitrary functions. \cite{L} also ingeniously obtained the form of the solutions to the $C_n$ and $B_n$ cases 
by enforcing symmetries among the integrands of the iterated integrals. (Also see \cite{EGR} for the nonabelian version.)

We verify that the elegant solutions of \cite{L} fit in our scheme, in the sense that they provide the general solutions to our conditions in Parts (1) of our theorems. The proof of this uses some basic properties of iterated integrals, which we present in Section \ref{ite int}. We also present a version for the $D_n$ case. As a whole, our paper can be said to have provided complete mathematical proofs for Leznov's solutions. 

We now list the formulas for the solutions to the $C_n$ and the $B_n$ cases, 
and for completeness also to the $A_n$ case. 

First some notation. For $m\geq 0$, let $F(x)=(f_0(x),\cdots,f_{m}(x))$ be a vector of $m+1$ functions of $x$. For $j\geq 0$, let 
$F^{(j)}(x)=(f_0^{(j)}(x),\cdots,f_{m}^{(j)}(x))$ be the $j$th derivative of $F(x)$ with respect to $x$. 
Similarly we have $G(y)=(g_0(y),\cdots,g_m(y))$ and $G^{(j)}(y)$. 

Let 
$$
F\cdot G=F(x)\cdot G(y)=\sum_{i=0}^{m} f_i(x)g_i(y)
$$
be the inner product. 
For $i\geq 1$, define 
\begin{equation}\label{taup}
\begin{split}
\tau_{i,F,G}(x,y)&=\det\Big(\big(\p_x^{j}\p_y^{k} (F\cdot G)\big)_{j,k=0}^{i-1}\Big)\\
&=\det\begin{pmatrix}
F\cdot G 
& F\cdot G' & \cdots & F\cdot G^{(i-1)}\\
F'\cdot G 
& F'\cdot G' & \cdots  & F'\cdot G^{(i-1)}\\
\cdots  & \cdots  & \cdots &\cdots\\
F^{(i-1)}\cdot G
& F^{(i-1)}\cdot G' & \cdots & F^{(i-1)}\cdot G^{(i-1)} 
\end{pmatrix}.
\end{split}
\end{equation}
For convenience, we also define $\tau_{0,F,G}(x,y)=1$. 
(When it is clear from the context what the function vectors $F(x)$ and $G(y)$ are, we just write $\tau_i(x,y)$ or even $\tau_i$.) 

To relate to the solutions to the Toda field theories in \cite{L}, we need 
iterated integrals. 
Given $m$ functions $\phi_1(x),\cdots,\phi_m(x)$ and for a sequence $(a_1,\cdots,a_k)$ with $1\leq a_i\leq m$, we define
\begin{multline}\label{short}
\ci(a_1\cdots a_k)
:=\int^x_0\phi_{a_1}(x_1)\,dx_1\int^{x_1}_0\phi_{a_2}(x_2)\,dx_2\cdots\int^{x_{k-1}}_0\phi_{a_k}(x_k)\,dx_k
\end{multline}
as a function of $x$. 
(For brevity, we usually omit the commas between the different entries of the sequence.) 
Most likely the sequence is monotonic or at least piecewise monotonic. We will use an arrow or some dots to denote a monotonic piece with the given initial and end values. For example, 
$$
\ci(1\cdots 4)=\ci(1234),\quad \ci(1\to 3\to 1)=\ci(12321).
$$
One fine point, important to this paper, is that some integrands may be repeated in the iterated integrals. We always write out the repetition explicitly. Therefore 
$$
\ci(3,3)=\int^x_0\phi_{3}(x_1)\,dx_1\int^{x_1}_0\phi_{3}(x_2)\,dx_2 \quad \text{but}\quad \ci(3\to 3)=\ci(3)=\int_0^x \phi_3(x_1)\,dx_1.
$$ 
We also use the convention that $\ci(\varnothing)=1$ for the empty sequence. In particular, $\ci(1\to 0)=1$ since, as a convention in this paper, there is no $\phi_0(x)$. 

For short, we will write $\displaystyle\int f$
for $\displaystyle\int^x_0 f(x_1)\,dx_1$. Note that 
\begin{equation}\label{diff i}
\ci(a_1\cdots a_k)=\int \phi_{a_1}\ci(a_2\cdots a_k)\text{\ \ and\ \ } \frac{d}{dx}\ci(a_1\cdots a_k)=\phi_{a_1}\ci(a_2\cdots a_k).
\end{equation}

\begin{theorem}[$A_n$, \cites{W-sym,L}]\label{thm-A_n} 
\begin{enumerate}
\item Let 
$$F(x)=(f_0(x),\cdots,f_{n}(x))$$
be a vector of $n+1$ functions of $x$ such that the Wronskian 
\begin{equation}\label{w=1}
W(F)=W(f_0,\cdots,f_{n})=1.
\end{equation}
Let $G(x)=(g_0(y),\cdots,g_{n}(y))$ be a vector of $n+1$ functions of $y$ with exactly the same property. 

Then the 
$$u_i=-\log \tau_i,\quad 1\leq i\leq n$$
satisfy the $A_n$ Toda field theory \er{toda}, where $\tau_i$ is defined in \er{taup}. 

\item Furthermore, let  $\phi_1(x),\cdots,\phi_n(x)$ be $n$ arbitrary functions of $x$. Define 
\begin{gather}
f_0(x)=\prod_{i=1}^n \phi_i(x)^{-\frac{n+1-i}{n+1}}=\phi_1(x)^{-\frac{n}{n+1}}\cdots \phi_n(x)^{-\frac{1}{n+1}},\label{def f_0}\\
f_i(x)=f_0(x)\ci(1\to i),\quad i\geq 1. \nm
\end{gather}
Then the most general solution to \er{w=1} is 
\begin{equation}\label{W=1}
F(x)=(f_0,f_1,f_2,\cdots,f_n)=f_0(1,\ci(1),\ci(12),\cdots,\ci(1\to n)).
\end{equation}
\end{enumerate}
\end{theorem}

\begin{theorem}[$C_n$]\label{thm-C_n} 
\begin{enumerate}
\item 
Let 
$$F(x)=(f_0(x),\cdots,f_{2n-1}(x))$$ be a vector of $2n$ functions of $x$, such that 
\begin{align}
C(F^{(i)},F^{(i+1)})&=0,\quad 0\leq i\leq n-2,\label{c=0}\\
C(F^{(n-1)},F^{(n)})&=-1,\label{c=1}
\end{align}
where $C(X,Y)=X\Omega Y^T$ is the bilinear form using the skew-symmetric
\begin{equation}\label{Omega}
\Omega=\begin{pmatrix}
0 & I_n\\
-I_n & 0
\end{pmatrix}
\end{equation}
preserved by the symplectic group $Sp(2n,\bc)$. 

Let $G(y)=(g_0(y),\cdots,g_{2n-1}(y))$ be a vector of $2n$ functions of $y$ with exactly the same properties. 

Then the 
$$u_i=-\log \tau_i,\quad 1\leq i\leq n,$$
satisfy the $C_n$ Toda field theory \er{toda}, where $\tau_i$ is defined in \er{taup} using the above $F(x)$ and $G(y)$. 

\item 
Furthermore, let  $\phi_1(x),\cdots,\phi_n(x)$ be $n$ arbitrary functions of $x$. Define 
\begin{align}
p(x)&=\frac{1}{\phi_1(x)\cdots\phi_{n-1}(x)\sqrt{\phi_n(x)}}&&\label{def p(x)}\\
f_i(x)&=(-1)^{n-i} p(x)\ci(1\to i), && 0\leq i\leq n-1\nm\\
f_i(x)&=p(x)\ci(1\to n\to (i-n+1)),&& n\leq i\leq 2n-2\nm\\
f_{2n-1}(x)&=p(x)\ci(1\to n)&&\nm
\end{align}
Then the most general solution to \er{c=0} and \er{c=1} is 
\begin{multline}\label{Lez Cn}
F(x)=(f_0,f_1,f_2,\cdots,f_{2n-1})\\
=p(x)\big((-1)^{n},(-1)^{n-1}\ci(1),(-1)^{n-2}\ci(12),\cdots,-\ci(1\to (n-1)),\\
\ci(1\to n\to 1),\ci(1\to n\to 2),\cdots,\ci(1\to n, (n-1)),\ci(1\to n)\big).
\end{multline}
\end{enumerate}
\end{theorem}

\begin{theorem}[$B_n$]\label{thm-B_n} 
\begin{enumerate}
\item Let 
$$F(x)=(f_0(x),\cdots,f_{2n}(x))$$ be a vector of $2n+1$ functions of $x$, 
such that 
\begin{align}
B(F^{(i)},F^{(i)})&=0,\quad 0\leq i\leq n-1,\label{b=0}\\
B(F^{(n)},F^{(n)})&=1,\label{b=1}
\end{align}
where $B(X,Y)=X\Theta Y^T$ is the bilinear form using the symmetric 
\begin{equation}\label{Theta}
\Theta=\begin{pmatrix}
0& I_n &0\\
I_n & 0 & 0\\
0 & 0 & 1
\end{pmatrix}\
\end{equation}
preserved by the orthogonal group $O(2n+1,\bc)$. 

Let $G(y)=(g_0(y),\cdots,g_{2n}(y))$ be a vector of $2n+1$ functions of $y$ with exactly the same properties. 

Then the 
\begin{equation}\label{B special}
\begin{split}
u_i&=-\log \tau_i, \quad 1\leq i\leq n-1,\\
u_n&=- \frac{1}{2}\log \tau_n 
\end{split}
\end{equation}
satisfy the $B_n$ Toda field theory \er{toda} (with a $\mp \frac{1}{2}$ as coefficient for the last equation for $u_{n,xy}$). 
\item 
Furthermore, let  $\phi_1(x),\cdots,\phi_n(x)$ be $n$ arbitrary functions of $x$. Define 
\begin{align*}
p(x)&=\frac{1}{\phi_1(x)\cdots\phi_{n-1}(x){\phi_n(x)}}&&\\
f_i(x)&=(-1)^{n-i} p(x)\ci(1\to i), && 0\leq i\leq n-1\\
f_i(x)&=p(x)\ci(1\to n,n\to (i-n+1)),&& n\leq i\leq 2n-2\\
f_{2n-1}(x)&=p(x)\ci(1\to n,n)&&\\
f_{2n}(x)&=p(x)\ci(1\to n)&&
\end{align*}
Then the most general solution to \er{b=0} and \er{b=1} is 
\begin{multline}\label{Lez Bn}
F(x)=(f_0,f_1,f_2,\cdots,f_{2n})\\
=p(x)\big((-1)^n,(-1)^{n-1}\ci(1),(-1)^{n-2}\ci(12),\cdots,-\ci(1\to (n-1)),\\
\ci(1\to n,n\to 1),\ci(1\to n,n\to 2),\cdots,\ci(1\to n,n,(n-1)),\ci(1\to n,n),\ci(1\to n)\big).
\end{multline}
\end{enumerate}
\end{theorem}

\begin{remark} Although the $F(x)$'s in \er{Lez Cn} and \er{Lez Bn} involve some particular orders and signs to satisfy our conditions, we note that they don't matter in $F\cdot G$ and hence in the $\tau_i$ since the $G(y)$'s will follow the same patterns. It is in this way that Leznov's solutions \cite{L} are presented. 
\end{remark}

The paper is organized as follows. In Section \ref{solve}, we study \er{find phi} in a DS-gauge for the $C_n$ and $B_n$ cases and obtain the solutions as minors. In Section \ref{ind int}, we present some algebraic properties of principal minors. In Section \ref{part 1}, we prove Parts (1) of the theorems using Section \ref{ind int}. In the $C_n$ and $B_n$ cases, we apply the Gram-Schmidt process to complete the vectors into a symplectic or orthogonal basis. In Section \ref{ite int}, we present  some differential properties of iterated integrals, including a version for the $D_n$ case. In Section \ref{part 2}, we prove Parts (2) of the theorems using Section \ref{ite int}, 
and we also say a little more about  the $D_n$ case. 

\medskip
\noindent{\bf Acknowledgment.} The author thanks Prof. Ian Anderson 
for 
many useful discussions, and Prof. L\'aszl\'o Feh\'er for some useful correspondences. He also thanks the referee for careful reading and detailed comments. 

\section{Solve chiral fields in DS-guage}\label{solve}

In this section, we solve \er{find phi} in the $C_n$ and $B_n$ groups $Sp(2n,\bc)$ and $O(2n+1,\bc)$ using some DS-gauges. Note that the corresponding preserved skew-symmetric and symmetric matrices are \er{Omega} and \er{Theta}. We follow \cite{FH} for choices of root vectors. 

In the $C_n$ case, we use $e_{-\a_i}=-E_{i+1,i}+E_{n+i,n+i+1}$ for $1\leq i\leq n-1$ and $e_{-\a_n}=E_{2n,n}$, where $E_{ij}$ is the matrix with a 1 at the $(i,j)$ position and zero everywhere else. Also $H_{\a_i}=E_{i,i}-E_{i+1,i+1}-E_{n+i,n+i}+E_{n+i+1,n+i+1}$ for $1\leq i\leq n-1$ and $H_{\a_n}=E_{n,n}-E_{2n,2n}$. 
We also choose the slice basis $s_j=E_{n-j+1,2n-j+1}$ for $1\leq j\leq n$. 

For concreteness,
we present the $C_3$ case. 
Writing \er{find phi} out in terms of one column vector $(\varphi_1,\cdots,\varphi_6)^T$ of $\Phi$, we have 
$$
\begin{pmatrix}
\varphi_1'\\
\varphi_2'\\
\varphi_3'\\
\varphi_4'\\
\varphi_5'\\
\varphi_6'
\end{pmatrix}=
\begin{pmatrix}
0 & 0 & 0 & I_3 & 0 & 0\\
-1 & 0 & 0 & 0 & I_2 & 0\\
0 & -1 & 0 & 0 & 0 & I_1\\
0 & 0 & 0 & 0 & 1 & 0\\
0 & 0 & 0 & 0 & 0 & 1\\
0 & 0 & 1 & 0 & 0 & 0
\end{pmatrix}
\begin{pmatrix}
\varphi_1\\
\varphi_2\\
\varphi_3\\
\varphi_4\\
\varphi_5\\
\varphi_6
\end{pmatrix}
=\begin{pmatrix}
I_3\varphi_4\\
-\varphi_1+I_2\varphi_5\\
-\varphi_2+I_1\varphi_6\\
\varphi_5\\
\varphi_6\\
\varphi_3
\end{pmatrix}
$$
Therefore regarding $\varphi_4(x)$ as free and writing it as $f(x)$, we have
\begin{align}
\varphi_4&=f\nm\\
\varphi_5&=\varphi_4'=f'\nm\\
\varphi_6&=\varphi_5'=f''\nm\\
\varphi_3&=\varphi_6'=f'''\nm\\
\varphi_2&=-\varphi_3'+I_1\varphi_6=-f^{(4)}+I_1 f''\label{phi2}\\
\varphi_1&=-\varphi_2'+I_2\varphi_5=f^{(5)}-I_1 f_1'''-I_1' f_1''
+I_2f'\label{phi1}
\end{align}
Actually the first equation $\varphi_1'=I_3\varphi_4$ gives
\begin{equation*}
f^{(6)}-I_1f^{(4)}-2I_1' f'''+(I_2-I_1'') f''+I_2' f'-I_3 f=0.
\end{equation*}
Putting 6 such columns together, we assume the solution to \er{find phi} in $C_3$ is 
\small
\begin{equation*}
\Phi(x)=\begin{pmatrix}
f_1^{(5)}-*
& f_2^{(5)}-* &  f_3^{(5)}-* &  f_4^{(5)}-* &  f_5^{(5)}-* &  f_6^{(5)}-*\\
-f_1^{(4)}+* & -f_2^{(4)}+* & -f_3^{(4)}+* & -f_4^{(4)}+* & -f_5^{(4)}+* & -f_6^{(4)}+*\\
f_1''' & f_2''' & f_3''' & f_4''' & f_5''' & f_6''' \\
f_1 & f_2 & f_3 & f_4 & f_5 & f_6\\
f_1' & f_2' & f_3' & f_4' & f_5' & f_6'\\
f_1'' & f_2'' & f_3'' & f_4'' & f_5'' & f_6''  
\end{pmatrix}
\end{equation*}
\normalsize
where the $*$'s are from \er{phi1} and \er{phi2} involving the $I$'s. 
This $\Phi(x)$ belongs to $Sp(6,\bc)$. 
Now the $I$'s are arbitrary local conservation laws and can be suitably adjusted for the $f$'s. We find the implication of $\Phi(x)\in Sp(6,\bc)$ on the $f$'s by concentrating on the last 4 rows in this case. 

More concretely, let $F(x)=(f_1(x),\cdots,f_6(x))$ be 
a vector of 6 functions of $x$. 
Then $F'''(x),F(x)$, $F'(x),F''(x)$ appear as the last 4 rows of a matrix $\Phi(x)$ in $Sp(6,\bc)$ under the conditions 
$$
C(F^{(i)},F^{(j)})=0\text{ for }0\leq i,j\leq 3 \text{ except }C(F''',F'')=1,
$$
where $C$ is defined in \er{Omega}. We will prove in the next section that the fewer conditions 
\er{c=0} and \er{c=1} in Theorem \ref{thm-C_n} Part (1) imply the above conditions. 

On the $y$-side, we have the transposed version $\Psi(y)$ for a vector function $G(y)=(g_1(y),\cdots,g_6(y))$. Therefore from \er{chirality}, we have
\begin{equation}\label{cn ex}
\Upsilon(x,y)
=\begin{pmatrix}
* & * & * & * & * & * \\
* & * & * & * & * & * \\
* & * 
& F'''\cdot G''' & F'''\cdot G & F'''\cdot G' & F'''\cdot G''\\ 
* & * 
& F\cdot G''' & F\cdot G & F\cdot G' & F\cdot G''\\
* & * 
& F'\cdot G''' & F'\cdot G & F'\cdot G' & F'\cdot G''\\
* & * 
 & F''\cdot G''' & F''\cdot G & F''\cdot G' & F''\cdot G''
\end{pmatrix},
\end{equation}
where the $*$'s are entries that contain the $I$'s for either $x$ or $y$. 

The solutions $u_i(x,y)$ should be some minors of $\Upsilon(x,y)$ invariant under some residual gauges. As explained in \cite{W-sym}, the residual gauges are 
$$
\Upsilon\mapsto \a\Upsilon\beta^{-1},\quad \text{where }\a=\a(x)\in N_+,\ \beta=\beta(y)\in N_-. 
$$
Here $N_+$ and $N_-$ are the positive and negative unipotent subgroups of $G=Sp(6,\bc)$ corresponding to $\fn_+$ and $\fn_-$. 

In our choice of basis, we see that $N_+$ is overall block-upper-triangular, but the the lower right block is lower triangular, and has 1's on the diagonal. On the contrary, $N_-$ is overall block-lower-triangular, but the the lower right block is upper triangular, and has 1's on the diagonal. Therefore the invariants under the residual gauges are those principal minors of $\Upsilon$ starting from position $(4,4)$ with increasing ranks going downward. Note that the entry $\Upsilon_{4,4}=F\cdot G$, and those bigger minors happen to be our $\tau_i$ for increasing $i$'s in \er{taup}, with the current $F(x)$ and $G(y)$. Comparison with 
$$
\exp\Big(\sum_{i=3}^n u_i H_{\a_i}\Big)=\text{Diag}\,(e^{u_1},e^{u_2-u_1},e^{u_3-u_2},e^{-u_1},e^{-u_2+u_1},e^{-u_3+u_2})
$$
gives $\tau_i=e^{-u_i}$ for $1\leq i\leq 3$. This gives our solutions in Theorem \ref{thm-C_n} Part (1). 
We will prove this directly in Section \ref{part 1}. 

In the $B_n$ case, we choose $e_{-\a_i}=-E_{i+1,i}+E_{n+i,n+i+1}$ for $1\leq i\leq n-1$ and $e_{-\a_n}=E_{2n,2n+1}-E_{2n+1,n}$. Also $H_{\a_i}=E_{i,i}-E_{i+1,i+1}-E_{n+i,n+i}+E_{n+i+1,n+i+1}$ for $1\leq i\leq n-1$ and $H_{\a_n}=2E_{n,n}-2E_{2n,2n}$. 
We also choose the slice basis $s_1=E_{n,2n+1}-E_{2n+1,2n}$ and $s_j=E_{n-j+1,2n-j+2}-E_{n-j+2,2n-j+1}$ for $2\leq j\leq n$. 

For concreteness, we present the $B_2$ case. 
Writing \er{find phi} out in terms of a column vector of $\Phi(x)$, we have 
$$
\begin{pmatrix}
\varphi_1'\\
\varphi_2'\\
\varphi_3'\\
\varphi_4'\\
\varphi_5'
\end{pmatrix}=
\begin{pmatrix}
0 & 0 & 0 & I_2 & 0\\
-1 & 0 & -I_2 & 0 & I_1\\
0 & 0 & 0 & 1 & 0\\
0 & 0 & 0 & 0 & 1\\
0 & -1 & 0 & -I_1 & 0 
\end{pmatrix}
\begin{pmatrix}
\varphi_1\\
\varphi_2\\
\varphi_3\\
\varphi_4\\
\varphi_5
\end{pmatrix}
=\begin{pmatrix}
I_2\varphi_4\\
-\varphi_1-I_2\varphi_3+I_1\varphi_5\\
\varphi_4\\
\varphi_5\\
-\varphi_2-I_1\varphi_4
\end{pmatrix}
$$
Therefore regarding $\varphi_3(x)$ as free and writing it as $f(x)$, we have
\begin{align}
\varphi_3&=f\nm\\
\varphi_4&=\varphi_3'=f'\nm\\
\varphi_5&=\varphi_4'=f''\label{last for B}\\
\varphi_2&=-\varphi_5'-I_1\varphi_4=-f'''-I_1 f'\nm\\
\varphi_1&=-\varphi_2'-I_2\varphi_3+I_1\varphi_5=f^{(4)}+2I_1 f''+I_1' f'-I_2 f\nm
\end{align}
and the first equation $\varphi_1'=I_2\varphi_4$ gives 
$$
f^{(5)}+2I_1f'''+3I_1'f''+(I_1''-2I_2)f'-I_2'f=0.
$$
Again we concentrate on the terms without the $I$'s. 
Let $F(x)=(f_1(x),\cdots,f_5(x))$ be a vector of 5 functions of $x$. Then $F(x),F'(x),F''(x)$ appear as the last 3 rows of a matrix in $O(5,\bc)$ under the conditions
$$
B(F^{(i)},F^{(j)})=0\text{ for }0\leq i,j\leq 2 \text{ except }B(F'',F'')=1,
$$
where $B$ is defined in \er{Theta}. We will prove in the next section that the fewer conditions 
\er{b=0} and \er{b=1} in Theorem \ref{thm-B_n} Part (1) imply the above conditions. 

Similarly in this case after incorporating the $\Psi(y)$, we have 
\begin{equation}\label{bn ex}
\Upsilon(x,y)
=\begin{pmatrix}
 * & * & * & * & * \\
 * & * & * & * & * \\
* & *  & F\cdot G & F\cdot G' & F\cdot G''\\
* & *  & F'\cdot G & F'\cdot G' & F'\cdot G''\\
* & *  & F''\cdot G & F''\cdot G' & F''\cdot G''
\end{pmatrix}.
\end{equation}

Again, the analysis of residual gauges tells us that the invariants are the the principal minors starting from position $(3,3)$ with increasing ranks, which are $\tau_i$'s \er{taup} in terms of our current $F(x)$ and $G(y)$. Comparison with 
$$
\exp\Big(\sum_{i=1}^2 u_i H_{\a_i}\Big)=\text{Diag}\,(e^{u_1},e^{2u_2-u_1},e^{-u_1},e^{-2u_2+u_1},1)
$$
gives $\tau_1=e^{-u_1}$ and $\tau_2=e^{-2u_2}$ and hence the solutions in Theorem \ref{thm-B_n} Part (1). Again this will be proved directly in Section \ref{part 1}. 

\section{Related algebraic properties}\label{ind int}

For their possible independent interest, we present these identities of principal minors of general linear, symplectic and orthogonal matrices, in this separate section. The author discovered these identities in his study of the Toda field theories for various Lie algebras. 

First the general linear case. Let $A\in GL(n,\bc)$ be a non-degenerate matrix. Let $S\subset \underline n:=\{1,2,\cdots,n\}$ be a subset of indices. Let $|S|$ denote the number of elements in $S$ and we often write $m$ for $|S|$. We also denote the complement of $S$ in $\underline n$ by $\bar S$. When needed, the enumeration of $S$ is written as $S=\{s_1,s_2,\cdots,s_m\}$ with $s_1<s_2<\cdots<s_m$. Let $M_S^A$ denote the principal minor of $A$ with indices in $S$. That is, if $A=(a_{ij})_{i,j=1}^n$, then $M^A_S=\det(a_{s_is_j})_{i,j=1}^m$. 

The efficient way to view minors is through exterior products. Let $\{e_i\}_{i=1}^n$ be the standard basis of $\bc^n$. Define 
$$
e_S:=e_{s_1}\wedge e_{s_2}\wedge \cdots \wedge e_{s_m}\in \wedge^m \bc^n.
$$

A matrix $A\in GL(n,\bc)$ defines a linear transformation of $\bc^n$, and naturally this induces a transformation of exterior powers $\wedge^* \bc^n$, for which we still use $A$ as the notation. Then the principal minor $M^A_S$ is nothing but the entry of the matrix for this induced transformation on $\wedge^m \bc^n$ in the $(e_S,e_S)$ position. Note that for the top exterior form, 
\begin{equation}\label{top}
A(e_1\wedge\cdots\wedge e_n)=(\det A)(e_1\wedge \cdots\wedge e_n).
\end{equation}

Let $C$ denote the cofactor matrix of $A$, that is, $C_{ij}$ is the minor of $A$ after deleting the $i$th row and the $j$th column, multiplied by $(-1)^{i+j}$. Note that when $A$ is invertible, $A^{-1}=\frac{1}{\det A}C^T$. 

\begin{proposition} \label{direct} For $A\in GL(n,\bc)$ with $C$ as its cofactor matrix and $S\subset \{1,\cdots,n\}$, we have 
\begin{equation*}
M^C_S=M^A_{\bar S}\cdot(\det A)^{|S|-1}.
\end{equation*}
\end{proposition}

\begin{proof} We have the following sequence of identities:
\begin{align*}
M^A_{\bar S}e_S\wedge e_{\bar S}&=e_S\wedge (M^A_{\bar S} e_{\bar S})=e_S\wedge (Ae_{\bar S})\\
&=A\big((A^{-1} e_S)\wedge e_{\bar S}\big)\overset\dag=(\det A)\bigg(\Big(\big(\frac{1}{\det A}C^{T}\big) e_{S}\Big)\wedge e_{\bar S}\bigg)\\
&=\frac{1}{(\det A)^{|S|-1}} M^C_S e_S\wedge e_{\bar S},
\end{align*}
where equality $\dag$ uses \er{top}. 
\end{proof}

\begin{remark}
The special case of Proposition \ref{direct} when $|S|=2$ is called the Jacobi identity \cite{LS-book}. 
\end{remark}

Now the symplectic and orthogonal cases. Recall $A\in C_n=Sp(2n,\bc)$ iff $A^T \Omega A=\Omega$ with $\Omega$ defined in \er{Omega}. 
Similarly $A\in O(n,\bc)$ iff $A^T\Theta A=\Theta$, where for $B_n=O(2n+1,\bc)$, $\Theta$ is given in \er{Theta}, and for $D_n=O(2n,\bc)$, $\Theta$ is defined as 
$
\Theta=\begin{pmatrix}
0& I_n\\
I_n & 0
\end{pmatrix}. 
$
Let $S\subset \{1,\cdots,2n(+1)\}$ be a subset. The following proposition relates the principal minors of $A$ itself, and we omit the superscript $A$ in the notation $M^A_S$ for simplicity. 
Let $\iota$ be the inversion of the first and second halves of indices, which would fix the last $2n+1$ in the $B_n$ case, that is, 
$$
\iota(k)=\begin{cases}
k+n & \text{if }1\leq k\leq n\\
k-n & \text{if }n+1\leq k\leq 2n\\
2n+1 &\text{if }k=2n+1
\end{cases}
$$

\begin{proposition}\label{swap} For $A\in Sp(2n,\bc)$, $A\in SO(2n+1,\bc)$ or $A\in SO(2n,\bc)$, and $S\subset \{1,\cdots,2n(+1)\}$, we have 
$$
M_S=M_{\iota(\bar S)},
$$
where $M_S$ stands for the principal minor of $A$ with indices in $S$, $\bar S$ is the complement of $S$, and $\iota(\bar S)$ is its image under the inversion $\iota$. 

In the orthogonal cases, if $\det A=-1$, then we have $M_S=-M_{\iota(\bar S)}$. 
\end{proposition}

\begin{proof} For $A\in C_n$, we have $A^{-1}=\Omega^{-1}A^T\Omega=-\Omega A^T\Omega$. For $A\in B_n$ or $D_n$, we have $A^{-1}=\Theta^{-1}A^T \Theta=\Theta A^T\Theta$. Now for any index set $I$, $\Omega e_I=\pm e_{\iota(I)}$ and $\Omega^2 e_I=-e_I$. Also $\Theta e_I=e_{\iota(I)}$ and $\Theta^2 e_I=e_I$. Below we show the details for the symplectic case, and the special orthogonal cases are even simpler. We have 
\begin{align*}
M_{S}e_S\wedge e_{\bar S}&=(M_{S}e_S)\wedge e_{\bar S}=(Ae_S)\wedge e_{\bar S}\\
&=A\big(e_S\wedge (A^{-1}e_{\bar S})\big)\overset{\ddagger}=e_{S}\wedge (-\Omega A^T\Omega)e_{\bar S}\\
&=e_S\wedge \big((-\Omega)A^T(\pm e_{\iota(\bar S)})\big)=e_S\wedge (-\Omega)(\pm M_{\iota(\bar S)} e_{\iota(\bar S)})\\
&=e_S\wedge (M_{\iota(\bar S)})e_{\bar S}=(M_{\iota(\bar S)}) e_S\wedge e_{\bar S}, 
\end{align*}
where the equality $\ddag$ uses \er{top} and that $\det A=1$ for $A\in Sp(2n,\bc)$ or $SO(n,\bc)$.

Note that when $\det A=-1$ in the orthogonal cases, we have the negative sign coming in at equality $\ddag$.
\end{proof}

\begin{remark} Without going into details, we note that Propositions \ref{direct} and \ref{swap} can be generalized to general minors, not necessarily principal. More specifically, let $S,T\subset \{1,\cdots,n\}$ be subsets with the same cardinality, then for $A$ nondegenerate, we have 
$$
M^C_{S,T}=\pm M_{\bar S,\bar T}^A \cdot (\det A)^{|S|-1}
$$
for a carefully determined sign. Similarly, let $S, T \subset \{1,\cdots,2n(+1)\}$ be subsets with the same cardinality, then for $A$ symplectic or orthogonal, we have 
$$
M_{S,T}=\pm M_{\iota(\bar S)\iota(\bar T)}.
$$
\end{remark}

\section{Parts (1) of Theorems}\label{part 1}

In this section, we first present a simple lemma about the solutions $u_i$ to \er{toda} in terms of some related functions, which we call $\sigma_i$. For completeness and the reader's convenience, we prove Proposition \ref{general} for general $\tau_i(x,y)$ in \er{taup} for any function vectors $F(x)$ and $G(y)$ using our Proposition \ref{direct}, and  we provide a proof of Theorem \ref{thm-A_n} Part (1).
Then using Proposition \ref{swap}, we also prove Parts (1) of Theorems \ref{thm-C_n} and \ref{thm-B_n}. The methods here are that we apply the Gram-Schmidt process to obtain a symplectic or orthogonal matrix under our conditions. At the end of this section, we add some remarks about taking care of different coefficients in \er{toda}. 

Define 
\begin{equation}\label{def tau}
\sigma_i:=e^{-u_i},\quad 1\leq i\leq n. 
\end{equation}
Also for a function $v=v(x,y)$, define 
\begin{equation}\label{def dd}
DD(v):=v\cdot v_{xy}-v_x\cdot v_y.
\end{equation}
\begin{lemma}[\cite{BBT}]\label{use tau} The $u_i$ satisfy \er{toda} if and only if
\begin{equation}\label{in terms of tau}
DD(\sigma_i)=\prod_{j\neq i}\sigma_j^{-a_{ij}}.
\end{equation}
\end{lemma}
\begin{proof}
By \er{def tau}, $u_i=-\log \sigma_i$. Then 
\begin{align}
u_{i,x}&=-\frac{\sigma_{i,x}}{\sigma_i},\nm\\
u_{i,xy}&=-\frac{\sigma_{i,xy}\sigma_i-\sigma_{i,x}\sigma_{i,y}}{\sigma_i^2}=-\frac{DD(\sigma_i)}{\sigma_i^2}\label{explain}
\end{align}
It is clear that \er{in terms of tau} holds if and only if 
$$
u_{i,xy}=-\exp\Big(2u_i+\sum_{j\neq i}a_{ij}u_j\Big)=-\exp\Big(\sum_{j=1}^n a_{ij}u_j\Big)
$$
since $a_{ii}=2$, which is \er{toda}. 
\end{proof}




\begin{proposition}\label{general} For the $\tau_i$ defined in \er{taup} for any function vectors $F(x)$ and $G(y)$ of the same length, we have 
\begin{equation*}
DD(\tau_i)=\tau_{i-1}\tau_{i+1},\quad i\geq 1. 
\end{equation*}
\end{proposition}

\begin{proof} It is clear that $DD(\tau_1)=\tau_2$ by definition.  Let 
$$T_i=\big(\p_x^{j}\p_y^{k} (F\cdot G)\big)_{j,k=0}^{i-1}$$
 be the matrix in \er{taup} such that $\tau_i=\det T_i$. Then for $2\leq i\leq n$,
we have\begin{align*}
T_{i+1}&=\begin{pmatrix}
T_i &  F^{(\leq i-1)}\cdot G^{(i)}\\
F^{(i)}\cdot G^{(\leq i-1)} &  F^{(i)}\cdot G^{(i)}
\end{pmatrix}\\
&=\begin{pmatrix}
T_{i-1} & F^{(\leq i-2)}\cdot G^{(i-1)}& F^{(\leq i-2)}\cdot G^{(i)} \\
F^{(i-1)}\cdot G^{(\leq i-2)} & F^{(i-1)}\cdot G^{(i-1)} & F^{(i-1)}\cdot G^{(i)}  \\
F^{(i)}\cdot G^{(\leq i-2)} & F^{(i)}\cdot G^{(i-1)}  & F^{(i)}\cdot G^{(i)} 
\end{pmatrix}.
\end{align*}
Let $C$ denote the cofactor matrix of $T_{i+1}$. We use indices for rows and columns from $0$ to $i$. Then it is clear that $\tau_i=C_{i,i}$. Through simple calculations, we see that $\tau_{i,x}=-C_{i-1,i}$, that is, the minor of $T_{i+1}$ after deleting the second last row and last column. Similarly $\tau_{i,y}=-C_{i,i-1}$ and $\tau_{i,xy}=C_{i-1,i-1}$. 

Therefore by Proposition \ref{direct}, we have 
\begin{align*}
DD(\tau_i)&=\tau_i\cdot \tau_{i,xy}-\tau_{i,x}\tau_{i,y}\\
&=C_{i,i}C_{i-1,i-1}-C_{i-1,i}C_{i,i-1}=M^C_{\{i-1,i\}}\\
&=\det(T_{i-1})\det(T_{i+1})^{2-1} =\tau_{i-1}\tau_{i+1}.
\end{align*}
\end{proof}


\begin{proof}[Proof of Theorem \ref{thm-A_n} Part (1)] The Cartan matrix for $A_n$ is 
$$
\begin{pmatrix}
2 & -1 & & &\\
-1 & 2 & -1 & & \\
 & \ddots & \ddots & \ddots& \\
 & &-1 & 2 & -1\\
 & & & -1 & 2
 \end{pmatrix}.
$$
In view of Lemma \ref{use tau}, we need to prove that the $\tau_i$ in \er{taup} for our current $F(x)$ and $G(y)$ satisfy 
\begin{equation}\label{tau A_n}
\begin{split}
DD(\tau_i)&=\tau_{i-1}\tau_{i+1},\quad 1\leq i\leq n-1\\
DD(\tau_n)&=\tau_{n-1}
\end{split}
\end{equation}

This follows directly from Proposition \ref{general}, as soon as we realize the following for the last one.  
That is, $\tau_{n+1}=W(F)W(G)$ as defined in \er{taup} and hence it is 1 by the Wronskian conditions \er{w=1}. 
\end{proof}



To apply Proposition \ref{swap} in the symplectic 
case, we first present a proposition. 
\begin{proposition}\label{Gram-Schmidt} Let $F(x)=(f_0(x),\cdots,f_{2n-1}(x))$ be a vector of $2n$ functions of $x$ satisfying \er{c=0} and \er{c=1}. Then the following $(n+1)\times 2n$ matrix
\begin{equation}\label{lower}
\begin{pmatrix}
F^{(n)}\\
F\\
F'\\
\vdots\\
F^{(n-1)}
\end{pmatrix}
\end{equation}
can be completed into a symplectic matrix $\Phi(x)\in Sp(2n, \bc)$ by adding $n-1$ rows on the top. 
\end{proposition}

\begin{proof}
We apply the Gram-Schmidt process in the symplectic case to find a symplectic matrix $\Phi$ using the vectors $F^{(i)}$, $0\leq i\leq 2n-1$, which in generic cases are lienarly independent. We write the row vectors of $\Phi$ by $\Phi^i$ for $1\leq i\leq 2n$. Then $\Phi$ is symplectic iff $C(\Phi^{i},\Phi^{n+i})=1=-C(\Phi^{n+i},\Phi^{i})$ for $1\leq i\leq n$, and all the other $C(\Phi^i,\Phi^j)=0$, where $C$ is defined in \er{Omega}. 

We let $\Phi^{n+i}=F^{(i-1)}$ for $1\leq i\leq n$ and $\Phi^{n}=F^{(n)}$ to contain the lower rows as specified in \er{lower}. We will show that this is legitimate by our conditions \er{c=0} and \er{c=1}. We will then define the $\Phi^{n-i}$ through the Gram-Schmidt process applied successively to the $F^{(n+i)}$ for $1\leq i\leq n-1$. 

Now define 
\begin{equation}\label{def I_j}
\t I_j=C(F^{(n+j-1)},F^{(n+j)}),\quad 1\leq j\leq n-1
\end{equation}
to be functions of $x$. 
We show that $2n-1$ conditions \er{c=0}, \er{c=1} and \er{def I_j} determine all the $C(i,j):=C(F^{(i)},F^{(j)})$ for $0\leq i\leq j\leq 2n-1$. 

Call $l:=i+j$ the level of $C(i,j)$. We run increasing induction on $l$
and decreasing induction on the first index $i$. 

Since
$
\frac{d}{dx}C(i,j)=C(i+1,j)+C(i,j+1), 
$
we have
\begin{equation}\label{induction}
C(i,j+1)=\frac{d}{dx}C(i,j)-C(i+1,j),
\end{equation}
where the two terms on the right have either a lower level or a bigger first index. Therefore we only need to know the leading term with the biggest $i$ at each level. When the level is odd, say $2k+1$, the leading term is $C(k,k+1)$, and these are defined precisely by our conditions \er{c=0}, \er{c=1} and \er{def I_j}. When the level is even, say $2k$, then the leading term $C(k,k)=0$ since $C$ is skew-symmetric. 

Furthermore, from the induction procedure \er{induction}, we have
\begin{align}
C(i,j)&=0, &&\text{if }0\leq i+j\leq 2n-2,\label{more c=0}\\
C(i,j)&=\pm 1, &&\text{if }i+j=2n-1,\label{more c=1}\\
C(i,j)&=0,&&\text{if }i+j=2n,\label{level 2n}
\end{align}
since in our defining conditions \er{c=0}, \er{c=1} and \er{def I_j}, the first non-zero $C(i,j)$ appear at level $2n-1$ as $-1$, and the leading term at level $2n$ is $C(n,n)=0$. 

Therefore it is legitimate for $\Phi$ to have the lower part \er{lower}. To illustrate the Gram-Schmidt process, let's consider the next row $\Phi^{n-1}$ using the vector function $F^{(n+1)}$. 

We know $C(F^{(n+1)},F^{(j)})=0$ for $0\leq j\leq n-3$ and $j=n-1$, and also $C(F^{(n+1)},F^{(n-2)})=-1$ by \er{more c=0}, \er{level 2n} and \er{more c=1}. Also $C(F^{(n+1)},F^{(n)})=-\t I_1$ by \er{def I_j}, therefore we have
$$
\Phi^{n-1}=-(F^{(n+1)}-\t I_1\cdot F^{(n-1)}),
$$
so that $C(\Phi^{n-1},\Phi^n)=C(\Phi^{n-1},F^{(n)})=0$ and $C(\Phi^{n-1},\Phi^{2n-1})=C(\Phi^{n-1},F^{(n-2)})\\=1$. 

This process can be continued to fill up the matrix $\Phi$ in terms of the $F^{(k)}$ and the $\t I_j$. 
\end{proof}

\begin{proof}[Proof of Theorem \ref{thm-C_n} Part (1)]
For $n\geq 2$, the Cartan matrix for $C_n$ is 
$$
\begin{pmatrix}
2 & -1 & & &\\
-1 & 2 & -1 & & \\
 & \ddots & \ddots & \ddots& \\
 & &-1 & 2 & -1\\
 & & & -2 & 2
 \end{pmatrix}
$$
Therefore by Lemma \ref{use tau}, we need to prove that the $\tau_i$ in \er{taup} for our function vectors $F(x)$ and $G(y)$ in Theorem \ref{thm-C_n} satisfy 
\begin{equation}\label{tau C_n}
\begin{split}
DD(\tau_i)&=\tau_{i-1}\tau_{i+1},\quad 1\leq i\leq n-1\\
DD(\tau_n)&=\tau_{n-1}^2.
\end{split}
\end{equation}

After Proposition \ref{general}, only the last equation needs a demonstration. 

Proposition \ref{Gram-Schmidt} shows that we have a symplectic matrix $\Phi(x)$ with the lower rows to be given by \er{lower}. The same can be done for the $y$-side to get a symplectic matrix $\Psi(y)$ with the last columns to be given by the $(G^{(j)})^T$ in a suitable order. Therefore $\Upsilon(x,y)=\Phi(x)\cdot \Psi(y)$ is a symplectic matrix with its lower-right corner to be given by the $F^{(i)}\cdot G^{(j)}$ in a suitable order. (See \er{cn ex} for an example.) 

By Proposition \ref{general}, we have $DD(\tau_n)=\tau_{n-1}\tau_{n+1}$. 
But $\tau_{n+1}$ is the principal minor $M_{\{n,n+1,\cdots,2n\}}$ for the symplectic matrix $\Upsilon=\Phi(x)\cdot \Psi(y)$. By Proposition \ref{swap}, this is equal to $M_{\{n+1,\cdots,2n-1\}}$, which is seen to be $\tau_{n-1}$. Therefore $\tau_{n+1}=\tau_{n-1}$, and the last equation in \er{tau C_n} is proved.  
\end{proof}



Very similarly, we can prove the $B_n$ case. 
\begin{proof}[Proof of Theorem \ref{thm-B_n} Part (1)] 
Assume $n\geq 2$. The Cartan matrix for $B_n$ is 
$$
\begin{pmatrix}
2 & -1 & & &\\
-1 & 2 & -1 & & \\
 & \ddots & \ddots & \ddots& \\
 & &-1 & 2 & -2\\
 & & & -1 & 2
 \end{pmatrix}
$$

According to our solution formula \er{B special} in this case, we let $\sigma_i=\tau_i$ for $1\leq i\leq n-1$ and $\sigma_n=\sqrt{\tau_n}$, with the $\tau_i$ defined in \er{taup} for our function vectors $F(x)$ and $G(y)$ in Theorem \ref{thm-B_n}. 
Therefore by Lemma \ref{use tau}, we need to prove that 
\begin{equation}\label{tau B_n}
\begin{split}
DD(\tau_i)&=\tau_{i-1}\tau_{i+1},\quad 1\leq i\leq n-2\\
DD(\tau_{n-1})&=\tau_{n-2}(\sqrt{\tau_n})^2\\
DD(\sqrt{\tau_n})&=\pm \frac{1}{2}\tau_{n-1}.
\end{split}
\end{equation}
(To accommodate the $\pm\frac{1}{2}$ in the last equation, see Remark \ref{coeff's}.) 

After Proposition \ref{general}, only the last equation needs a demonstration. 

Similarly to Proposition \ref{Gram-Schmidt}, we  first complete the matrix with the last rows as $F,\cdots,F^{(n)}$ to an orthogonal one. 
We define 
\begin{equation}\label{def I_j-B}
\t I_j=B(F^{(n+j)},F^{(n+j)}),\quad 1\leq j\leq n
\end{equation}
to be functions of $x$. 
Then similarly to the proof of Proposition \ref{Gram-Schmidt}, the conditions \er{b=0}, \er{b=1} and \er{def I_j-B} determine all the $B(i,j):=B(F^{(i)},F^{(j)})$ for $0\leq i\leq j\leq 2n$, and we can run the Gram-Schmid procedure to complete the orthogonal matrix.  
We omit the details, but just point out that $B(i,i+1)=\frac{1}{2}\frac{d}{dx} B(i,i)$ since $B$ is symmetric and this is how the leading terms for the odd levels are determined. 
Also note that the determinant of the orthogonal matrix can be $\pm 1$. 

The same can be done for the $y$-side to get an orthogonal matrix $\Psi(y)$ with its last columns to be 
$(G^{(j)})^T$ for $0\leq j\leq n$. 
Therefore $\Upsilon(x,y)=\Phi(x)\cdot \Psi(y)$ is an orthogonal matrix with its lower-right corner to be given by $F^{(i)}\cdot G^{(j)}$. (See \er{bn ex} for an example.) $\Upsilon(x,y)$ has determinant $\pm 1$. 

From \er{def dd} and by simple computation, we have 
$$
DD(\sqrt{v})=\frac{1}{2}\frac{DD(v)}{v}.
$$
Therefore again by Proposition \ref{general},
$$
DD(\sqrt{\tau_n})=\frac{1}{2}\frac{DD(\tau_{n})}{\tau_n}=\frac{1}{2}\frac{\tau_{n-1}\tau_{n+1}}{\tau_n}.
$$
For the orthogonal matrix $\Upsilon(x,y)$, 
$$\tau_{n+1}=M_{\{n+1,\cdots,2n,2n+1\}}\text{ and }\tau_{n}=M_{\{n+1,\cdots,2n\}}.$$
By Proposition \ref{swap}, $\tau_{n+1}=\pm \tau_{n}$ determined by the sign of $\det(\Upsilon(x,y))$. Therefore 
$$
DD(\sqrt{\tau_n})=\pm\frac{1}{2}\tau_{n-1}.
$$
\end{proof}

\begin{remark} There are direct proofs of the relations $\tau_{n+1}=\tau_{n-1}$ for the $C_n$ case and $\tau_{n+1}=\pm \tau_n$ for the $B_n$ case, without using the Gram-Schmidt process and Proposition \ref{swap}. Rather the classical relations (see, for example, \cite{FH}*{Appendix F}) between determinants
and the pairings $C(\cdot,\cdot)$ and $B(\cdot,\cdot)$ in \er{Omega} and \er{Theta} are used. The procedure is long, and we won't present the details here. 
\end{remark}

\begin{remark}\label{coeff's} The negative Toda field theory \er{toda} is related to the positive version 
$$
v_{i,xy}=\exp\Big({\sum_{j=1}^n a_{ij}v_j}\Big),\quad 1\leq i\leq n,
$$
through a simple linear transformation 
$$
v_i=u_i+\theta_i,
$$
where the $\theta_i$ 
satisfy the equation 
$$
A\cdot (\theta_1,\cdots,\theta_n)^T=((2m_1+1)i\pi,\cdots,(2m_n+1)i\pi)^T.
$$
Here $A$ is the Cartan matrix \er{Cartan}, and the $m_i$ are arbitrary integers coming from the multiple-valuedness of the exponential function. Since $A$ is invertible, such $\theta_i$ always exist, and actually there are infinitely many of them due to the $m_i$. 

This remark also applies to the $B_n$ Toda field theory where,  in view of \er{tau B_n} and \er{explain}, we have solved a variant
\begin{align*}
w_{i,xy}&=-\exp\Big({\sum_{j=1}^n a_{ij}w_j}\Big),\quad 1\leq i\leq n-1,\\
w_{n,xy}&=\mp\frac{1}{2}\exp\Big({\sum_{j=1}^n a_{nj}w_j}\Big)=\mp\frac{1}{2}\exp (-w_{n-1}+2w_n)
\end{align*}
To relate this to the \er{toda}, we only need to let $w_i=u_i+r_i$, where the $r_i$ satisfy, for example,  
$$
A\cdot (r_1,\cdots,r_n)^T=(0,\cdots,0,\pm\log 2)^T. 
$$
\end{remark}

\section{Related differential properties}\label{ite int}

Because of its possible independent interest and for clarity of exposition, we present in this section some relations among iterated integrals and their derivatives that we need later to prove Parts (2) of our theorems. 

For completeness, we prove Part (2) of Theorem \ref{thm-A_n} first. 
\begin{lemma}[\cite{L}]\label{prove W=1} \er{W=1} satisfy \er{w=1}. 
\end{lemma}

\begin{proof} 
The Wronskian matrix 
is $\big(f_j^{(i)}\big)_{i,j=0}^{n}$, and we take the ranges of row and column indices to be from 0 to $n$. We first think of $f_0$ as just an arbitrary function of $x$, without remembering its definition \er{def f_0} in terms of the $\phi$'s. We write each derivative as 
$$
f_j^{(i)}=\sum_{k=0}^{\min(i,j)} f_j^{(i)}(k). 
$$
Here $f_j^{(i)}(k)$ is the summand of terms of $f_j^{(i)}$ where $\phi_k$ or its derivatives appear as factors outside the integral, but none of the $\phi_l$ for $l>k$ does. It is clear that $f_j^{(i)}(k)$ is non-zero only if $k\leq \min(i,j)$. 

For example, by \er{diff i}, we have 
\begin{align*}
f_3&=f_0\ci(123)\\
f_3'&=f_0'\ci(123)+f_0\phi_1\ci(23)\\
f_3''&=f_0''\ci(123)+2f_0'\phi_1\ci(23)+f_0\phi_1'\ci(23)+f_0\phi_1\phi_2\ci(3).
\end{align*}
Then
\begin{align*}
f_3^{(2)}(0)&=f_0''\ci(123)\\
f_3^{(2)}(1)&=2f_0'\phi_1\ci(23)+f_0\phi_1'\ci(23)\\
f_3^{(2)}(2)&=f_0\phi_1\phi_2\ci(3).
\end{align*}

Then inspections show that 
\begin{align*}
f_j^{(i)}(k)&=f_k^{(i)}(k)\cdot \ci((k+1)\to j),&& k\leq \min(i,j),\\
f_k^{(k)}(k)&=f_0\phi_1\cdots\phi_k, &&k\geq 0.
\end{align*}

Therefore the following successive column operations, replacing the old columns by the new ones, 
\begin{align*}
C_j&\mapsto C_j-\ci(1\to j)C_0, && j\geq 1\\
C_j&\mapsto C_j-\ci(2\to j)C_1, && j\geq 2\\
& \cdots &&\cdots\\
C_j&\mapsto C_j-\ci((n-1)\to j)C_{n-2},&&j\geq n-1\\
C_n&\mapsto C_n-\ci(n)C_{n-1} &&
\end{align*}
transform the Wronskian into a lower triangular matrix with diagonal entries
$$f_0,f_0\phi_1,f_0\phi_1\phi_2,\cdots,f_0\phi_1\phi_2\cdots\phi_n.$$
Thus the determinant $W(F)$ is 1, by the definition of $f_0$ in \er{def f_0}. 
\end{proof} 

To obtain results for the $B_n$ and $C_n$ type Lie algebras, we need more symmetry properties of the iterated integrals. 

\begin{lemma}\label{prod rule} We have 
\begin{align*}
\ci(a_1\cdots a_l)\ci(b_1\cdots b_m)&=\int \phi_{a_1}\,\ci(a_2\cdots a_l)\ci(b_1\cdots b_m)\\
&\quad +\int \phi_{b_1}\,\ci(a_1\cdots a_l)\ci(b_2\cdots b_m).
\end{align*}
\end{lemma}

\begin{proof} This follows from the product rule and \er{diff i}, since our integrals are fixed with lower limit 0. 
\end{proof}

\begin{proposition}\label{concrete} Let $n\geq 1$ be an integer, and let $\phi_1(x),\cdots,\phi_n(x)$ be $n$ functions of $x$. Define $\ci(a_1\cdots a_k)$ as in \er{short}. Then we have 
\begin{equation}\label{source}
\begin{split}
\sum_{i=0}^n (-1)^i \,\ci(1\to i)\,\ci(n\to i+1)=0.
\end{split}
\end{equation}
\end{proposition}

\begin{proof} We denote the left hand side of \er{source} by $A(n)$, and call it an alternating sum of $n$ integrands $\phi_1,\cdots,\phi_n$. 
We prove that $A(n)=0$ by induction. The identity is trivial when $1$. 

Let's assume that $A(n-1)=0$ for $n-1$ arbitrary functions $\phi_i(x)$. Then by Lemma \ref{prod rule} and \er{diff i}, 
\begin{align*}
A(n)&=\ci(n\to 1)-\ci(n\to 2)\ci(1)+\cdots+(-1)^{n-1}\ci(n)\ci(1\to (n-1))+(-1)^n\ci(1\to n)\\
&=\ci(n\to 1)-\bigg(\int \phi_n \ci((n-1)\to 2)\ci(1)+\int \phi_1 \ci(n\to 2)\bigg)+\cdots\\
&\quad +(-1)^{n-1}\bigg(\int \phi_n \ci(1\to (n-1))+\int \phi_1 \ci(n)\ci(2\to (n-1))\bigg)+(-1)^n\ci(1\to n)\\
&=\int \phi_n \Big[\ci((n-1)\to 1)-\ci((n-1)\to 2)\ci(1)+\cdots+(-1)^{n-1}\ci(1\to (n-1))\Big]\\
&\quad -\int \phi_1 \Big[\ci(n\to 2)-\cdots+(-1)^{n-2}\ci(n)\ci(2\to (n-1))+(-1)^{n-1}\ci(2\to n)\Big]\\
&=\int \phi_n A(n-1)-\int \phi_1 \tilde A(n-1)\\
&=0,
\end{align*}
where $\tilde A(n-1)$ is the $A(n-1)$ with the $n-1$ integrands being $\phi_2,\cdots,\phi_n$. 
\end{proof}

We will later use the convention that the LHS of \er{source} is 1 when $n=0$. 

\begin{remark} Lemma \ref{prod rule} and Proposition \ref{concrete} are related to the shuffle relations for iterated integrals in \cite{Chen}. 
\end{remark}

Proposition \ref{concrete} can be rephrased and strengthened as follows. 
\begin{proposition}\label{strong} Let $n\geq 2$ be an integer. Let  
\begin{equation}\label{def j}
J=J_{n+1}=\begin{pmatrix}
& & & & 1\\
& & & -1 & \\
& & 1 & & \\
& \iddots & & & \\
(-1)^{n} & & & &\\
\end{pmatrix}
\end{equation}
be a matrix of rank $n+1$ with alternating $\pm 1$'s on the skew diagonal. Note that $J$ is symmetric if $n$ is even, and skew-symmetric if $n$ is odd. Let $\phi_1(x),\cdots,\phi_{n}(x)$ be $n$ functions of $x$. Let 
\begin{equation}\label{def F}
\tilde F(x)=(1,\ci(1),\ci(12),\cdots,\ci(1\to n))
\end{equation}
be a vector of $n+1$ functions. Also for $0\leq i\leq n$, define a delayed version of $\t F(x)$ by 
\begin{equation}\label{dif}
\delta_i\t F(x)=(\underbrace{0,\cdots,0}_{i},1,\ci((i+1)),\cdots,\ci((i+1)\to n)).
\end{equation}
Define the ``swap" function 
\begin{equation}\label{def swap}
{s}(i)=
n+1-i,\quad 1\leq i\leq n
\end{equation}
on the indices $\{1,\cdots,n\}$ of the $\phi$'s. 
This defines the ``swap" on iterated integrals by $s\ci(a_1\cdots a_m)=\ci(s(a_1)\cdots s(a_m))$, and also on a vector of iterated integrals. 

Then we have
\begin{align}
(\delta_i\tilde F)J(s\delta_j\tilde F)^T&=0, &&\text{if }i+j\neq n\label{neq n}\\
(\delta_i\tilde F)J(s\delta_j\tilde F)^T&=(-1)^i, &&\text{if }i+j= n.\label{equ n}
\end{align}
\end{proposition}

\begin{proof} It is clear that for vectors $X=(x_0,x_1,\cdots,x_n)$ and $Y=(y_0,y_1,\cdots,y_n)$, we have
$$
XJY^T=x_0y_n-x_1y_{n-1}+\cdots+(-1)^n x_ny_0.
$$
By definition, we have
$$
s\delta_j\tilde F=(\underbrace{0,\cdots,0}_{j},1,\ci((n-j)),\cdots,\ci((n-j)\to 1)).
$$

When $i+j>n$, $(\delta_i\tilde F)J(s\delta_j\tilde F)^T$ is obviously zero, since there are too many zeros. When $i+j=n$, we have $(\delta_i\tilde F)J(s\delta_j\tilde F)^T=(-1)^i$ as the signed product of the two ones. 

When $i+j<n$, we have
\begin{align*}
(\delta_i\tilde F)J(s\delta_j\tilde F)^T&=(-1)^i \ci((n-j)\to (i+1))+(-1)^{i+1}\ci(i+1)\ci((n-j)\to i)+\cdots\\
&\quad +(-1)^{n-j}\ci((i+1)\to (n-j))\\
&=(-1)^{i}A(n-i-j),
\end{align*}
where $A(n-i-j)$ is the $A$ in the proof of Proposition \ref{concrete} with $n-i-j$ integrands $\phi_{i+1},\cdots,\phi_{n-j}$. This is zero by Proposition \ref{concrete}. 
\end{proof} 


Now we use Proposition \ref{strong} to derive some differential relations. 

\begin{proposition}\label{general dr} Continue with the notation in Proposition \ref{strong}. 
Also let 
\begin{equation}\label{missed sqrt}
f_0(x)=\big(\prod_{i=1}^{n}\phi_i(x)\big)^{-\frac{1}{2}}.
\end{equation}
Define
\begin{equation}\label{old F}
F(x)=f_0(x)\tilde F(x)=f_0(x)(1,\ci(1),\cdots,\ci(1\to n)). 
\end{equation}
Recall that $F^{(i)}$ is the $i$th derivative of $F$. Then we have 
\begin{align}
F^{(i)}J(sF^{(j)})^T&=0, && \text{if }  i+j< n,\label{we get all}\\
F^{(i)}J(sF^{(j)})^T&=(-1)^i, && \text{if }  i+j= n.\label{combined}
\end{align}
\end{proposition}

\begin{proof} By \er{dif} and \er{diff i}, we have 
$$
\frac{d}{dx}\big(\delta_i\t F(x)\big)=\phi_{i+1}\big(\delta_{i+1}\t F(x)\big),\quad i\geq 0. 
$$
Therefore we see that, through a quick induction, the $i$th derivative
$$
F^{(i)}=\sum_{l=0}^{i-1} c_{i,l} (\delta_l\t F)+f_0\Big(\prod_{l=1}^i \phi_l\Big)\delta_i\t F,
$$
where the $c_{i,l}$ are some functions of $x$ in terms of $f_0,\phi_1,\cdots,\phi_{l}$ and their derivatives. 

Note that
$$
sF^{(j)}=(sF)^{(j)}=\sum_{k=0}^{j-1} s(c_{j,k}) (s\delta_k\t F)+f_0\Big(\prod_{k=n-j+1}^n \phi_{k}\Big)s\delta_j\t F.
$$

Therefore $F^{(i)}J(sF^{(j)})^T$ is a linear combination of $(\delta_l \t F)J(s\delta_k \t F)^T$ for $l\leq i$ and $k\leq j$ with the coefficients as some functions of $x$. When $i+j<n$, then all such $k+l<n$, and from \er{neq n} we have \er{we get all}. 

When $i+j=n$, from \er{neq n} and \er{equ n} we see the only nontrivial terms is
\begin{equation*}
\Big(f_0\big(\prod_{l=1}^i \phi_l)\Big) \Big(f_0\big(\prod_{k=n-j+1}^n \phi_{k})\Big) (\delta_i\t F)J(s\delta_j\t F)
=f_0^2\Big(\prod_{l=1}^n \phi_l\Big) (-1)^i=(-1)^i,
\end{equation*}
by the definition of $f_0$ in \er{missed sqrt}. 
\end{proof}

There are  corresponding versions of Propositions \ref{strong} and \ref{general dr} for the $D_n$ case, which we present here for completeness.
\begin{proposition}\label{alg for D} Let $n\geq 3$ be an integer. Let  
$$
K=\begin{pmatrix}
 & & & & & & & 1\\
  & & & && &-1 & \\
 & & & & &\iddots & & \\
 & & & &(-1)^{n-1} &&  & \\
 & & & (-1)^{n+1} & & & & \\
 & &\iddots & & &  & &\\
  & -1& & & & & & \\
1 & & & &  & & &
\end{pmatrix}
$$
be an symmetric matrix of rank $2n$. Let $\phi_1(x),\cdots,\phi_{2n-2}(x)$ be $2n-2$ functions of $x$. Let 
\begin{multline}
\t F(x)=(1,\ci(1),\ci(12),\cdots,\ci(1\to (n-2)),\ci(1\to (n-2)(n-1))\\
\ci(1\to (n-2) n),\ci(1\to (n-2)(n-1)n)+\ci(1\to (n-2)n(n-1)),\\
\ci(1\to (n-2)(n-1)n(n+1))+\ci(1\to (n-2)n(n-1)(n+1),\cdots,\\
\ci(1\to (2n-2))+\ci(1\to (n-2)n(n-1)(n+1)\to (2n-2))
\end{multline}
be a vector of $2n$ functions. Note that $\t F(x)$ goes like before for the first $(n-1)$ terms, then it branches using $\phi_{n-1}$ and $\phi_n$, and starting from the $(n+2)$nd term it is always symmetrized between $\phi_{n-1}$ and $\phi_n$. 
Define the ``swap" function 
$$
{s}(i)=\begin{cases}
2n-1-i & i\leq n-2\\
n-1 & i=n-1\\
n & i=n\\
2n-1-i & i\geq n+1
\end{cases}
$$
and let it act on iterated integrals and vectors of them as before. 
Then we have
$$
\t FK(s\t F)^T=0.
$$

Moreover, let
\begin{multline}\label{did}
\delta_i \t F(x)=(\underbrace{0,\cdots,0}_{i},1,\ci(i+1),\cdots,\ci((i+1)\cdots (n-2)(n-1)),\ci((i+1)\cdots (n-2)n),\\
\text{symmetrized terms})\qquad  i\leq n-2\\
\delta_{n-1}\t F(x)=(\underbrace{0,\cdots,0}_{n-1},1,0,\ci(n),\ci(n(n+1)),\cdots,\ci(n\cdots(2n-2)) \\
\delta_{n}\t F(x)=(\underbrace{0,\cdots,0}_{n},1,\ci(n-1),\ci((n-1)(n+1)),\cdots,\ci((n-1)(n+1)\cdots(2n-2)) \\
\delta_{i}\t F(x)=(\underbrace{0,\cdots,0}_{i},1,\ci(i+1),\cdots,\ci((i+1)\cdots(2n-2)) \hfill i\geq n+1.
\end{multline}
Then 
\begin{align*}
(\delta_i\tilde F)K(s\delta_j\tilde F)^T&=0, &&\text{if }i+j\neq 2n-1\\
(\delta_{i}\tilde F)K(s\delta_j\tilde F)^T&=(-1)^{i}, &&\text{if }i+j= 2n-1,\text{ and }i<j
\end{align*}
\end{proposition}

\begin{proof} When $i+j>2n-1$, $(\delta_i\tilde F)K(s\delta_j\tilde F)^T$ is obviously zero since there are too many zeros. When $i+j=2n$, we have $(\delta_i\tilde F)K(s\delta_j\tilde F)^T=(-1)^i$ when $i<j$ as the signed product of the two ones. 

Like in the proof of Proposition \ref{strong}, when $i+j<2n$ we have that $(\delta_i\tilde F)K(s\delta_j\tilde F)^T$ is one or two $A$'s from Proposition \ref{concrete} for a suitable sequence of integrand functions possibly with sign. For example when $i=j=0$, we have that $\t FK(s\t F)^T=A_1(2n-2)+A_2(2n-2)$, where $A_1(2n-2)$ is the $A$ in Proposition \ref{concrete} for the sequence $(1\to (2n-2))$ and $A_2(2n-2)$ for $(1\to (n-2)n(n-1)(n+1)\to (2n-2))$. They are both zero by \er{source}. 
\end{proof}

\begin{proposition} 
Continue with the notation in Proposition \ref{alg for D}. 
Let $f_0(x)=\Big(\prod_{i=1}^{2n-2}\phi_i(x)\Big)^{-\frac{1}{2}}$. Define
$$
F(x)=f_0(x)\tilde F(x)
$$
Then we have 
\begin{align}\label{dif for d}
F^{(i)}K(sF^{(j)})^T&=0, \qquad\qquad  0\leq i,j\leq n-1\text{ and }i+j<2n-2\\
F^{(n-1)}K(sF^{(n-1)})^T&=(-1)^{n-1}2.\label{a 2}
\end{align}
\end{proposition}

\begin{remark} Note that \er{a 2} is compatible with solving for a $\Phi(x)\in O(2n,\bc)$ in Section \ref{solve}, since one would have $f^{(n-1)}=\varphi_{n}+\varphi_{2n}$, if $\varphi_{n+1}=f$ in the $D_n$ case. (Compare with \er{last for B}.) Then $B(F^{(n-1)},F^{(n-1)})=2B(\Phi^{n},\Phi^{2n})=2$, where $\Phi^n$ and $\Phi^{2n}$ are the $n$th and $2n$th rows of the solution matrix $\Phi\in O(2n,\bc)$. 
\end{remark}

\begin{proof} By \er{did}, we have 
\begin{align*}
\frac{d}{dx}\delta_i\t F&=\phi_{i+1}\delta_{i+1}\t F, &&i\leq n-3\\
\frac{d}{dx}\delta_{n-2}\t F&=\phi_{n-1}\delta_{n-1}\t F+\phi_{n}\delta_n \t F\\
\frac{d}{dx}\delta_{n-1}\t F&=\phi_{n}\delta_{n+1}\t F\\
\frac{d}{dx}\delta_{n}\t F&=\phi_{n-1}\delta_{n+1}\t F. 
\end{align*}
Therefore 
\begin{align*}
F^{(i)}&=\sum_{l=0}^{i-1} c_{i,l} \delta_l\t F+f_0\Big(\prod_{l=1}^i \phi_l\Big)\delta_i\t F&&i\leq n-2\\
F^{(n-1)}&=\sum_{l=0}^{n-2} c_{n-1,l} \delta_l\t F+f_0\Big(\prod_{l=1}^{n-2}\phi_l\Big) (\phi_{n-1}\delta_{n-1}\t F+\phi_{n}\delta_n \t F) && \\
sF^{(j)}&=\sum_{k=0}^{j-1} s(c_{j,k}) s\delta_k\t F+f_0\Big(\prod_{k=2n-1-j}^{2n-2} \phi_k\Big)s\delta_j\t F&&j\leq n-2\\
sF^{(n-1)}&=\sum_{l=0}^{n-2} s(c_{n-1,l}) s\delta_l\t F+f_0\Big(\prod_{l=n+1}^{2n-2}\phi_l\Big) (\phi_{n-1}s\delta_{n-1}\t F+\phi_{n}s\delta_n \t F) &&
\end{align*}
By Proposition \ref{alg for D}, \er{dif for d} is easy to see, and also 
\begin{multline*}
F^{(n-1)}K(sF^{(n-1)})^T\\
=f_0^2\Big(\prod_{1\leq l\leq n-2}^{n+1\leq l\leq 2n-2} \phi_l\Big)(\phi_{n-1}\phi_n)\Big((\delta_{n-1}\t F)K(s\delta_{n}\t F)^T+(\delta_n\t F)K(s\delta_{n-1}\t F)^T\Big)\\
=f_0^2\Big(\prod_{l=1}^{2n-2} \phi_l\Big)2(-1)^{n-1}=(-1)^{n-1}2.
\end{multline*}
by the definition of $f_0$. 
\end{proof}

\section{Parts (2) of Theorems}\label{part 2}

The proofs of Parts (2) of Theorems \ref{thm-C_n} and  \ref{thm-B_n} use the results of Section \ref{ite int}, in particular Proposition \ref{general dr}, while the proof of Theorem \ref{thm-A_n} Part (2) is done in Lemma \ref{prove W=1}. 
\begin{proof}[Proof of Theorem \ref{thm-C_n} Part (2)] In Proposition \ref{general dr}, let the number of $\phi_i(x)$, called $n$ there, be $2n-1$. Also require that 
$$
\phi_{2n-i}(x)=\phi_i(x),\quad 1\leq i\leq n-1.
$$
Note that this is Leznov's \cite{L} ingenious idea to enforce symmetry. Then $f_0(x)$ in \er{missed sqrt} becomes $p(x)$ in \er{def p(x)}. Denote the corresponding $F(x)$ in \er{old F} by $\bar F(x)$, that is,
$$
\bar F(x)=p(x)(1,\ci(1),\cdots,\ci(1\to n),\ci(1\to n,(n-1)),\cdots,\ci(1\to n\to 1)).
$$
It is clear that $s\bar F=\bar F$ and also for all the derivatives $\bar F^{(i)}$, since the swap as in \er{def swap}
 in this case is $s(i)=2n-i$. 
 
It is easy to see that the $F(x)$ in \er{Lez Cn} is just $F(x)=\bar F(x) Q_{2n}$. Here
\begin{gather}\label{transf}
\quad Q_{2n}=\begin{pmatrix}
(-1)^{n} & & & & & & & \\
 & (-1)^{n-1}& & & & & & \\
 & &\ddots & & & & & \\
 & & & -1 & & & & \\
 & & & & & & &1 \\
 & & & & & &1 & \\
 & & & & &\iddots & & \\
 & & & &1 & & & \\
\end{pmatrix}\implies
Q_{2n}\Omega Q_{2n}^T=(-1)^n  J_{2n},
\end{gather}
where $\Omega$ is from \er{Omega} and $J_{2n}$ is as in \er{def j} with rank $2n$. Now
\begin{multline*}
C(F^{(i)},F^{(j)})=F^{(i)}\Omega (F^{(j)})^T=\bar F^{(i)}Q_{2n}\Omega Q_{2n}^T (\bar F^{(j)})^T\\
=(-1)^n \bar F^{(i)}J_{2n}(\bar F^{(j)})^T=\begin{cases}
0 & \text{if }i+j<2n-1\\
-1 & \text{if }(i,j)=(n-1,n)
\end{cases}
\end{multline*}
by Proposition \ref{general dr}. 
\end{proof}

\begin{proof}[Proof of Theorem \ref{thm-B_n} Part (2)] Again, this is very similar to the previous proof. We let the number of $\phi_i(x)$ in Proposition \ref{general dr} be $2n$, and we require that 
$$
\phi_{2n+1-i}(x)=\phi_i(x),\quad 1\leq i\leq n.
$$
Call the corresponding function vector by $\bar F(x)$, and the $F(x)$ in \er{Lez Bn} is $F(x)=\bar F(x)Q_{2n+1}$, where $Q_{2n+1}$ is as in \er{transf} but with the lower right block of rank $n+1$. Note that $Q_{2n+1}\Theta Q_{2n+1}^T=(-1)^n  J_{2n+1}$, for $\Theta$ in \er{Theta} and $J_{2n+1}$ as in \er{def j}. We omit the other details. 
\end{proof}

The solutions to $D_n$ Toda field theory are also explained in \cite{L}, although there is a typo. We record the result below for completeness. 
\begin{theorem}[$D_n$]\cite{L} Let $\phi_1(x),\cdots,\phi_n(x)$ be $n$ functions of $x$. Define
$$f_0(x)=\frac{1}{\phi_1\cdots\phi_{n-2}\sqrt{\phi_{n-1}}\sqrt{\phi_{n}}}.$$
Let 
\begin{multline}
F(x)=f_0(x)(1,\ci(1),\ci(12),\cdots,\ci(1\to (n-2)),\ci(1\to (n-2)(n-1))\\
\ci(1\to (n-2) n),\ci(1\to (n-2)(n-1)n)+\ci(1\to (n-2)n(n-1)),\\
\ci(1\to (n-2)(n-1)n(n-2))+\ci(1\to (n-2)n(n-1)(n-2),\cdots,\\
\ci(1\to n,(n-2)\to 1)+\ci(1\to (n-2),n\to 1)
\end{multline}
be a vector of $2n$ functions of $x$. 

Similarly for $n$ functions $\psi_1(y),\cdots,\psi_n(y)$ of $y$, define $G(y)$. 

Use the old definition of $\tau_i$ in terms of $F$ and $G$ as in \er{taup}. Then 
\begin{align}
\sigma_i&=\tau_i &&1\leq i\leq n-2\label{as b4}\\
\sigma_{n-1}&=\frac{\sqrt{\tau_n+2\tau_{n-1}}+\sqrt{\tau_n-2\tau_{n-1}}}{2}&&\nm\\
\sigma_{n}&=\frac{\sqrt{\tau_n+2\tau_{n-1}}-\sqrt{\tau_n-2\tau_{n-1}}}{2}\nm&&
\end{align}
are solutions to the $D_n$ Toda field theories \er{in terms of tau}. Here $\sigma_{n-1}$ and $\sigma_n$ are solutions to the following two conditions
\begin{gather}
\sigma_{n-1}\sigma_n=\tau_{n-1},\label{must do 1}\\
\sigma_{n-1}^2+\sigma_{n}^2=\tau_n.\label{must do 2}
\end{gather}
Since the Cartan matrix for $D_n$ is 
$$
\begin{pmatrix}
2 & -1 & & & & \\
-1 & 2 & -1 & & & \\
 & \ddots & \ddots & \ddots& & \\
 & &-1 & 2 & -1 & -1\\
 & & & -1 & 2 & \\
 & & & -1 &  & 2\\
 \end{pmatrix},
$$
by Proposition \ref{use tau} this asserts that 
\begin{align}
DD(\sigma_i)&=\sigma_{i-1}\sigma_{i+1},&&1\leq i\leq n-3\nm\\
DD(\sigma_{n-2})&=\sigma_{n-3}\sigma_{n-1}\sigma_n&&\label{s(n-2)}\\
DD(\sigma_{n-1})&=\sigma_{n-2}\label{s(n-1)}\\
DD(\sigma_{n})&=\sigma_{n-2}\label{sn}
\end{align}

Equivalently, the $u_i=-\log \sigma_i$ for $1\leq i\leq n$ are solutions to the $D_n$ Toda field theory in \er{toda}. 
\end{theorem}

\begin{remark} We finally remark that only the proofs of the last two equations \er{s(n-1)} and \er{sn} are not clear, although one can argue that \er{must do 1} and \er{must do 2} are the only ways to obtain a solution. This must be how Leznov \cite{L} arrived at these solutions and we reproduce the process as follows. By \er{as b4} and Proposition \ref{general}, we have
$$
DD(\sigma_{n-2})=DD(\tau_{n-2})=\tau_{n-3}\tau_{n-1}=\sigma_{n-3}\tau_{n-1}.
$$
Comparison with the wanted equation \er{s(n-2)} gives \er{must do 1}. A simply calculation from \er{def dd} shows that 
$$
DD(vw)=DD(v)w^2+v^2DD(w)
$$
for two functions $v$ and $w$ of $x$ and $y$. Therefore \er{must do 1} gives
$$
DD(\tau_{n-1})=DD(\sigma_{n-1}\sigma_{n})=DD(\sigma_{n-1})\sigma_{n}^2+\sigma_{n-1}^2 DD(\sigma_{n}).
$$
Proposition \ref{general} again gives that $LHS=\tau_{n-2}\tau_n$. The wanted equations \er{s(n-1)} and \er{sn} give that $RHS=\sigma_{n-2}(\sigma_{n-1}^2+\sigma_n^2)$. Therefore with \er{as b4}, we get \er{must do 2}. 
\end{remark}

\bigskip

\begin{bibdiv}
\begin{biblist}
\bib{BBT}{book}{
   author={Babelon, Olivier},
   author={Bernard, Denis},
   author={Talon, Michel},
   title={Introduction to classical integrable systems},
   series={Cambridge Monographs on Mathematical Physics},
   publisher={Cambridge University Press},
   place={Cambridge},
   date={2003},
   pages={xii+602},
   isbn={0-521-82267-X},
   review={\MR{1995460 (2004e:37085)}},
   doi={10.1017/CBO9780511535024},
}

\bib{W-sym}{article}{
   author={Balog, J.},
   author={Feh{\'e}r, L.},
   author={O'Raifeartaigh, L.},
   author={Forg{\'a}cs, P.},
   author={Wipf, A.},
   title={Toda theory and $\scr W$-algebra from a gauged WZNW point of view},
   journal={Ann. Physics},
   volume={203},
   date={1990},
   number={1},
   pages={76--136},
   issn={0003-4916},
   review={\MR{1071398 (92b:58084)}},
}

\bib{Chen}{article}{
   author={Chen, Kuo Tsai},
   title={Iterated path integrals},
   journal={Bull. Amer. Math. Soc.},
   volume={83},
   date={1977},
   number={5},
   pages={831--879},
   issn={0002-9904},
   review={\MR{0454968 (56 \#13210)}},
}
		
\bib{DS}{article}{
   author={Drinfel{\cprime}d, V. G.},
   author={Sokolov, V. V.},
   title={Lie algebras and equations of Korteweg-de Vries type},
   language={Russian},
   conference={
      title={Current problems in mathematics, Vol. 24},
   },
   book={
      series={Itogi Nauki i Tekhniki},
      publisher={Akad. Nauk SSSR Vsesoyuz. Inst. Nauchn. i Tekhn. Inform.},
      place={Moscow},
   },
   date={1984},
   pages={81--180},
   review={\MR{760998 (86h:58071)}},
}

\bib{EGR}{article}{
   author={Etingof, Pavel},
   author={Gelfand, Israel},
   author={Retakh, Vladimir},
   title={Factorization of differential operators, quasideterminants, and
   nonabelian Toda field equations},
   journal={Math. Res. Lett.},
   volume={4},
   date={1997},
   number={2-3},
   pages={413--425},
   issn={1073-2780},
   review={\MR{1453071 (98d:58081)}},
}

\bib{feher-early}{article}{
   author={Forg{\'a}cs, P.},
   author={Wipf, A.},
   author={Balog, J.},
   author={Feh{\'e}r, L.},
   author={O'Raifeartaigh, L.},
   title={Liouville and Toda theories as conformally reduced WZNW theories},
   journal={Phys. Lett. B},
   volume={227},
   date={1989},
   number={2},
   pages={214--220},
   issn={0370-2693},
   review={\MR{1013515 (91b:81046)}},
   doi={10.1016/S0370-2693(89)80025-5},
}


\bib{FH}{book}{
   author={Fulton, William},
   author={Harris, Joe},
   title={Representation theory},
   series={Graduate Texts in Mathematics},
   volume={129},
   note={A first course;
   Readings in Mathematics},
   publisher={Springer-Verlag},
   place={New York},
   date={1991},
   pages={xvi+551},
   isbn={0-387-97527-6},
   isbn={0-387-97495-4},
   review={\MR{1153249 (93a:20069)}},
}


\bib{K1}{article}{
   author={Kostant, Bertram},
   title={The principal three-dimensional subgroup and the Betti numbers of
   a complex simple Lie group},
   journal={Amer. J. Math.},
   volume={81},
   date={1959},
   pages={973--1032},
   issn={0002-9327},
   review={\MR{0114875 (22 \#5693)}},
}

\bib{L}{article}{
   author={Leznov, A. N.},
   title={On complete integrability of a nonlinear system of partial
   differential equations in two-dimensional space},
   language={Russian, with English summary},
   journal={Teoret. Mat. Fiz.},
   volume={42},
   date={1980},
   number={3},
   pages={343--349},
   issn={0564-6162},
   review={\MR{569026 (82e:58047)}},
}

\bib{LS}{article}{
   author={Leznov, A. N.},
   author={Saveliev, M. V.},
   title={Representation of zero curvature for the system of nonlinear
   partial differential equations $x_{\alpha ,z\bar z}={\rm
   exp}(kx)_{\alpha }$ and its integrability},
   journal={Lett. Math. Phys.},
   volume={3},
   date={1979},
   number={6},
   pages={489--494},
   issn={0377-9017},
   review={\MR{555332 (82k:58045)}},
   doi={10.1007/BF00401930},
}


\bib{LS-book}{book}{
   author={Leznov, A. N.},
   author={Saveliev, M. V.},
   title={Group-theoretical methods for integration of nonlinear dynamical
   systems},
   series={Progress in Physics},
   volume={15},
   note={Translated and revised from the Russian;
   Translated by D. A. Leuites},
   publisher={Birkh\"auser Verlag},
   place={Basel},
   date={1992},
   pages={xviii+290},
   isbn={3-7643-2615-8},
   review={\MR{1166385 (93d:58194)}},
}

\bib{N1}{article}{
   author={Nie, Zhaohu},
   title={Characteristic integrals for Toda field theories},
   journal={Preprint, Utah State University},
   date={2012},
}

\end{biblist}
\end{bibdiv}
\end{document}